\documentclass[11pt,letterpaper]{article}

\usepackage{amsmath,amsfonts,amsthm,amssymb,stmaryrd,graphicx,relsize,bm}  
\theoremstyle{plain}

\numberwithin{equation}{section}

\newtheorem{thm}{Theorem}[section]
\newtheorem{lem}[thm]{Lemma}

\newenvironment{exam}[1]
{\begin{flushleft}\textbf{Example #1}.\enspace}%
{\end{flushleft}}

\allowdisplaybreaks  

\newcounter{cond}

\newcommand{\trace}{tr}
\newcommand{\instr}{In}
\newcommand{\ob}{Ob}
\newcommand{\sob}{Sob}

\newcommand{\escript}{\mathcal{E}}
\newcommand{\fscript}{\mathcal{F}}
\newcommand{\hscript}{\mathcal{H}}
\newcommand{\iscript}{\mathcal{I}}
\newcommand{\jscript}{\mathcal{J}}
\newcommand{\lscript}{\mathcal{L}}
\newcommand{\oscript}{\mathcal{O}}

\newcommand{\sscript}{\mathcal{S}}
\newcommand{\uscript}{\mathcal{U}}
\newcommand{\vscript}{\mathcal{V}}

\newcommand{\iscripthat}{\widehat{\iscript}}
\newcommand{\jscripthat}{\widehat{\jscript}}
\newcommand{\lscripthat}{\widehat{\lscript}}
\newcommand{\iscriptbar}{\overline{\iscript}}
\newcommand{\jscriptbar}{\overline{\jscript}}

\newcommand{\brac}[1]{\left\{#1\right\}}
\newcommand{\paren}[1]{\left(#1\right)}
\newcommand{\sqbrac}[1]{\left[#1\right]}
\newcommand{\elbows}[1]{{\left\langle#1\right\rangle}}
\newcommand{\ket}[1]{{\left|#1\right>}}
\newcommand{\bra}[1]{{\left<#1\right|}}

\begin{document}

\title{DUAL QUANTUM INSTRUMENTS \\ AND SUB-OBSERVABLES}
\author{Stan Gudder\\ Department of Mathematics\\
University of Denver\\ Denver, Colorado 80208\\
sgudder@du.edu}
\date{}
\maketitle

\begin{abstract}
We introduce the concepts of dual instruments and sub-observables. We show that although a dual instruments measures a unique observable, it determines many sub-observables. We define a unique minimal extension of a sub-observable to an observable and consider sequential products and conditioning of sub-observables. Sub-observable effect algebras are characterized and studied. Moreover, the convexity of these effect algebras is considered. The sequential product of instruments is discussed. These concepts are illustrated with many examples of instruments. In particular, we discuss L\"uders, Holero and constant state instruments. Various conjectures for future research are presented.
\end{abstract}

\section{Introduction}  
In this section we only present general ideas and the detailed definitions will be given in Section~2. An instrument $\iscript$ is considered to be a two-step measurement process. In the first step, an input state $\rho$ is selected and a measurement of $\iscript$ is performed. The outcome $\Delta$ of this measurement is observed and the probability of this outcome is given by the trace $\trace\sqbrac{\iscript (\Delta )(\rho )}$. In the second step, the state $\rho$ is updated to a new state $\iscript (\Delta )(\rho )^\sim$ depending on the outcome $\Delta$ of the first step. When $\iscript$ produces the outcome $\Delta$ we say that the resulting effect is $A(\Delta )$. We call the map $\Delta\mapsto A(\Delta )$ an effect-valued measure or observable and say that $\iscript$ measures the observable $A$. As we shall see in Section~2, the probability distribution of $A$ in the state $\rho$ becomes:
\begin{equation*}
\trace\sqbrac{\rho A(\Delta )}=\trace\sqbrac{\iscript (\Delta )\rho}
\end{equation*}
In Section~2, we introduce the concept of a sub-observable $A_1$ which can be considered as a deficient observable in the sense that all the possible values of $A_1$ need not be attainable.

We also present the concept of a dual instrument $\iscript ^*$ in Section~2. As we shall see, $\iscript ^*$ satisfies the equation
\begin{align}                
\label{eq11}
\trace\sqbrac{\rho\iscript ^*(\Delta )(a)}=\trace\sqbrac{\iscript (\Delta )(\rho )a}
\end{align}
for states $\rho$, outcomes $\Delta$ and effects $a$. Equation~\eqref{eq11} gives a duality between $\iscript$ and $\iscript ^*$. Defining
$\iscript _a^* (\Delta )=\iscript ^*(\Delta )(a)$, we shall show in Section~2 that $\iscript _a^*$ is a sub-observable which we say is determined by $\iscript$. Thus, although $\iscript$ measures a unique observable, it determines many sub-observables.

In Section~3, we show that every sub-observable has a unique minimal extension to an observable. We also consider sequential products and conditioning of
sub-observables. We next characterize and study sub-observable effect algebras. In particular, we show that the set of sub-observables $\iscript _a^*$ determined by
$\iscript$ forms an effect algebra in the natural way. Moreover, the convexity of these effect algebras is considered. Section~4 studies sequential products of instruments. All of these concepts are illustrated with many examples of instruments and observables. In particular, we discuss L\"uders, Holevo and constant state instruments. Various conjectures for future research are presented

\section{Basic Definitions}  
Let $S$ be a quantum system described by a complex Hilbert space $H$ and $\lscript (H)$ be the set of bounded linear operators on $H$. For
$A,B\in\lscript (H)$, we write $A\le B$ if $\elbows{\phi ,A\phi}\le\elbows{\phi ,B\phi}$ for all $\phi\in H$. We call $a\in\lscript (H)$ an \textit{effect} if $0\le a\le I$ where $0$, $I$ are the zero and identity operators, respectively. An effect describes $a$ two-valued true-false experiment and the set of effects is denoted by $\escript (H)$.
If an $a\in\escript (H)$ is true, then its \textit{complement} $a'=I-a$ is false. Let $(\Omega _A,\fscript _A)$ be a measurable space. An \textit{observable} with
\textit{outcome space} $(\Omega _A,\fscript _A)$ is a normalized effect-valued measure $A\colon\fscript _A\to\escript (H)$. That is, $\Delta\mapsto A(\Delta )$ is countably additive in the strong operator topology and $A(\Omega _A)=I$. We interpret $A(\Delta )$ as the effect that is true when a measurement of $A$ results in an outcome in $\Delta$. A \textit{state} for the system $S$ is an effect $\rho$ that satisfies $\trace (\rho )=1$. States describe the initial condition of the system and the set of states is denoted by $\sscript (H)$ \cite{bgl95,dl70,hz12,kra83,nc00}. If $\rho\in\sscript (H)$ and $A$ is an observable, the \textit{distribution} of $A$ in the state
$\rho$ is a probability measure given by
\begin{equation*}
\Phi _\rho ^A(\Delta )=\trace\sqbrac{\rho A(\Delta )}
\end{equation*}
We denote the set of observables by $\ob (H)$. If $\Delta\mapsto A(\Delta )$ is countably additive but $A(\Omega _A)$ need not be $I$, then $A$ is called a
\textit{sub-observable}. We denote the set of sub-observables by $\sob (H)$.

An effect $\rho$ is called a \textit{partial state} if $\trace (\rho )\le 1$. An \textit{operation} is a completely positive linear map $\oscript\colon\lscript (H)\to\lscript (H)$ such that $\trace\sqbrac{\oscript (\rho )}\le 1$ for all $\rho\in\sscript (H)$. If $\trace\sqbrac{\oscript (\rho )}=1$ for all $\rho\in\sscript (H)$ then $\oscript$ is called a
\textit{channel} \cite{bgl95,hz12,kra83,nc00}. We denote the set of operations by $\oscript (H)$. Let $(\Omega _\iscript ,\fscript _\iscript )$ be a measurable space. An instrument with \textit{outcome space} $(\Omega _\iscript ,\fscript _\iscript )$ is a normalized operation-valued measure $\iscript$. That is, for $\Delta\in\fscript _\iscript$, $\Delta\mapsto\iscript (\Delta )\in\oscript (H)$ is countably additive in the strong operation topology and $\iscriptbar =\iscript (\Omega _\iscript )$ is a channel
\cite{hz12,nc00}. We denote the set of instruments by $\instr (H)$. We interpret $\iscript\in\instr (H)$ as an apparatus that measures an observable $\iscripthat$ and updates states. If $\rho\in\sscript (H)$, then $\iscript (\Delta )(\rho )$ is a partial state and assuming $\iscript (\Delta )(\rho )\ne 0$, then
\begin{equation*}
\sqbrac{\iscript (\Delta )(\rho )}^\sim =\frac{1}{\trace\sqbrac{\iscript (\Delta )(\rho )}}\,\iscript (\Delta )(\rho )
\end{equation*}
is a state. We call $\sqbrac{\iscript (\Delta )(\rho )}^\sim$ the \textit{update} of $\rho$ given that a measurement of $\iscript$ results in an outcome in $\Delta$. If
$\rho\in\sscript (H)$, the \textit{distribution} of $\iscript\in\instr (H)$ is the probability measure given by
\begin{equation*}
\Phi _\rho ^\iscript (\Delta )=\trace\sqbrac{\iscript (\Delta )(\rho )}
\end{equation*}

The \textit{dual instrument} to $\iscript\in\instr (H)$ is the unique map $\iscript ^*(\Delta )\colon\lscript (H)\to\lscript (H)$, $\Delta\in\fscript _\iscript$, satisfying
\begin{equation}                
\label{eq21}
\trace\sqbrac{\rho\iscript ^*(\Delta )A}=\trace\sqbrac{\iscript (\Delta )(\rho )A}
\end{equation}
for all $A\in\lscript (H)$ \cite{gud22}. It follows that $\Delta\mapsto\iscript ^*(\Delta )$ is countable additive and $\iscript ^*(\Delta )$ is a completely positive linear map such that $\iscript ^*(\Delta )\colon\escript (H)\to\escript (H)$ for all $\Delta\in\fscript _\iscript$ and $\iscript ^*(\Omega _\iscript )I=I$ \cite{gud22}. If $a\in\escript (H)$, we define $\iscript _a^*(\Delta )=\iscript ^*(\Delta )(a)$. Then $\iscript _a^*\in\sob (H)$ and we say that $\iscript$ \textit{determines} $\iscript _a^*$. We interpret 
$\iscript _a^*(\Delta )\in\escript (H)$ as the update of the effect $a$ given that a measurement of $\iscript$ results in an outcome in $\Delta$. In this way, an instrument not only updates states, it also can be employed to update effects. For all $\rho\in\sscript (H)$ we have from \eqref{eq21} that
\begin{equation*}
\trace\sqbrac{\rho\iscript_a^*(\Delta )}=\trace\sqbrac{\rho\iscript ^*(\Delta )a}=\trace\sqbrac{\iscript (\Delta )(\rho )a}
   =\trace\sqbrac{\iscript (\Delta )(\rho )}\trace\sqbrac{\paren{\iscript (\Delta )(\rho )}^\sim a}
\end{equation*}
Thus, the probability of the effect $\iscript _a^*(\Delta )$ when $S$ is in the state $\rho$ is the probability that a measurement of $\iscript$ results in an outcome in
$\Delta$ times the probability of $a$ in the updated state $\paren{\iscript (\Delta )(\rho )}^\sim$. Notice that 
\begin{equation*}
\iscript _I^*(\Omega _\iscript)=\iscript ^*(\Omega _\iscript )(I)=I
\end{equation*}
so $\iscript _I^*$ is an observable. We call $\iscript _I^*$ the \textit{observable measured} by $\iscript$ and we write $\iscripthat=\iscript _I^*$. Notice that $\iscripthat$ is the unique observable satisfying
\begin{equation*}
\trace\sqbrac{\rho\iscripthat (\Delta )}=\trace\sqbrac{\iscript (\Delta )(\rho )}
\end{equation*}
for all $\Delta\in\fscript _\iscript$, $\rho\in\sscript (H)$. We have that $\iscript _a^*$ is an observable if and only if
\begin{equation*}
\iscript _a^*(\Omega _\iscript )=\iscript ^*(\Omega _\iscript )(a)=I
\end{equation*}
which is equivalent to
\begin{equation}                
\label{eq22}
\trace\sqbrac{\iscript (\Omega _\iscript )(\rho )a}=\trace\sqbrac{\rho\iscript ^*(\Omega _\iscript )a}=\trace (\rho I)=1
\end{equation}

We now present some examples that illustrate the previous definitions. An observable $A$ is \textit{finite} if $\Omega _A$ is a finite set. In this case we assume that
$\fscript _A=2^{\Omega _A}$ so we need not specify the $\sigma$-algebra $\fscript _A$. We then write $A=\brac{a_x\colon x\in\Omega _A}$ and we have that
\begin{equation*}
A(\Delta )=\sum _{x\in\Delta}a_x
\end{equation*}
for all $\Delta\subseteq\Omega _A$. Corresponding to a finite observable $A$ we have the \textit{L\"uders instrument} $\lscript _x(\rho )=a_x^{1/2}\rho a_x^{1/2}$ for all $\rho\in\sscript (H)$, $x\in\Omega _A$ \cite{lud51}. It follows that
\begin{equation*}
\lscript (\Delta )(\rho )=\sum _{x\in\Delta}\lscript _x(\rho )=\sum _{x\in\Delta}a_x^{1/2}\rho a_x^{1/2}
\end{equation*}
for all $\rho\in\sscript (H)$, $\Delta\subseteq\Omega _A=\Omega _\lscript$. The dual instrument satisfies $\lscript _x^*(b)=a^{1/2}ba^{1/2}$ for all $x\in\Omega _A$,
$b\in\escript (H)$ \cite{gud22}. The sub-observables determined by $\lscript$ have the form $\lscript _b^*$, $b\in\escript (H)$ where
\begin{equation*}
\lscript _b^*(\Delta )=\sum _{x\in\Delta}a_x^{1/2}ba_x^{1/2}
\end{equation*}
If $\lscript _b^*\in\ob (h)$, we have that
\begin{equation*}
I=\lscript _b^*(\Omega _A)=\sum _{x\in\Omega _A}a_x^{1/2}ba_x^{1/2}
\end{equation*}
It follows that $b=I$. Hence,
\begin{equation*}
\lscripthat (\Delta )=\lscript _I^*(\Delta )=\sum _{x\in\Delta}a_x=A(\Delta )
\end{equation*}
for all $\Delta\subseteq\Omega _A$ so $\lscripthat =A$ is the only observable determined by $\lscript$.

A \textit{Holevo instrument with state} $\alpha$ and \textit{observable} $A$ has the form
\begin{equation*}
\hscript _{(\alpha ,A)}(\Delta )(\rho )=\trace\sqbrac{\rho A(\Delta )}\alpha
\end{equation*}
for all $\rho\in\sscript (H)$, $\Delta\in\fscript _A$ \cite{hol94}. The sub-observables determined by $\hscript _{(\alpha ,A)}$ become
\begin{equation*}
(\hscript _{(\alpha ,A)}^*)_a(\Delta )=\hscript _{(\alpha ,A)}^*(\Delta )(a)=\trace (\alpha a)A(\Delta )
\end{equation*}
Then $(\hscript {(\alpha ,A)}^*)_a$ is an observable if and only if
\begin{equation*}
\trace (\alpha a)A(\Omega _A)=\trace (\alpha a)I=I
\end{equation*}
This is equivalent to $\trace (\alpha a)=1$ which is equivalent to $(\hscript _{(\alpha ,A)}^*)_a(\Delta )=A(\Delta )$. Thus, $A=(\hscript _{(\alpha ,A)}^*)_I$ is the only observable determined by $\hscript _{(\alpha ,A)}$.

Let $A=\brac{a_x\colon x\in\Omega _A}$ be a finite observable and $\brac{\alpha _x\colon x\in\Omega _A}\subseteq\sscript (H)$. A
\textit{finite Holevo instrument} with \textit{states} $\brac{\alpha _x\colon x\in\Omega _A}$ and \textit{observable} $A$ has the form
\begin{equation*}
\hscript _{(\alpha ,A)}(x)(\rho )=\trace (\rho a_x)\alpha _x
\end{equation*}
The sub-observables determined by $\hscript _{(\alpha ,A)}$ becomes
\begin{equation*}
(\hscript _{(\alpha ,A)}^*)_a(x)=\hscript _{(\alpha ,A)}^*(x)(a)=\trace (\alpha _xa)a_x
\end{equation*}
We see that $(\hscript _{(\alpha ,A)}^*)_a$ is an observable if and only if
\begin{equation*}
(\hscript _{(\alpha ,A)}^*)_a(\Omega _A)=\sum _{x\in\Omega _A}\trace (\alpha _xa)a_x=I
\end{equation*}
which is equivalent to $\trace (\alpha _xa)=1$ for all $x\in\Omega _A$. Again, $A$ is the only observable determined by $\hscript _{(\alpha ,A)}$.
$\hscript _{(\alpha ,A)}$ is also called a \textit{conditional state preparator} \cite{hz12}.

A \textit{constant-state} instrument has the form $\iscript _\alpha (\Delta )(\rho )=\iscript (\Delta )(\alpha )$ where $\alpha\in\sscript (H)$, $\iscript\in\instr (H)$
\cite{gud22}. Since
\begin{equation*}
\trace\sqbrac{\rho\iscript _\alpha ^*(\Delta )(a)}=\trace\sqbrac{\iscript _\alpha (\Delta )(\rho )a}=\trace\sqbrac{\iscript (\Delta )(\alpha )a}
\end{equation*}
for all $\rho\in\sscript (H)$, we conclude that $\iscript _\alpha ^*(\Delta )(a)=\trace\sqbrac{\iscript (\Delta )(\alpha )a}I$ for all $\Delta\in\fscript _\iscript$. Hence, for all $a\in\escript (H)$ we obtain
\begin{equation*}
(\iscript _\alpha ^*)_a(\Delta )=\iscript _\alpha ^*(\Delta )(a)=\trace\sqbrac{\iscript (\Delta )(\rho )a}I
\end{equation*}
It follows that $(\iscript _\alpha ^*)_a$ is an observable if and only if
\begin{equation}                
\label{eq23}
\trace\sqbrac{\iscriptbar (\alpha )a}=1
\end{equation}
There can be many $a\in\escript (H)$ that satisfy \eqref{eq23}. For example, suppose $\iscriptbar (\rho )=\sum P_i\rho P_i$ where $P_i=\ket{\psi _i}\bra{\psi _i}$ and
$\brac{\psi _i}$ is an orthonormal basis for $H$. Also, suppose $\alpha =\ket{\psi _1}\bra{\psi _1}$. Then $\iscriptbar (\alpha )=P_1$ and we have
\begin{equation*}
\trace\sqbrac{\iscriptbar (\alpha )a}=\trace\sqbrac{\ket{a\psi _1}\bra{\psi _1}}=\elbows{\psi _1,a\psi _1}
\end{equation*}
Then $\elbows{\psi _1,a\psi _1}=1$ if and only if $a\psi _1=\psi _1$ and there are many $a\in\escript (H)$ that satisfy this. We conclude that $\iscript _\alpha$ can determine many sub-observables. The unique observable measured by $\iscript _\alpha$ is $\iscripthat _\alpha =(\iscript _\alpha ^*)_I$ where
\begin{equation*}
\iscripthat _\alpha (\Delta )=\trace\sqbrac{\iscript (\Delta )(\alpha )}I
\end{equation*}

We now discuss the algebraic structure of $\escript (H)$. An \textit{effect algebra} is a four-tuple $(E,0,1,\oplus )$ where $E$ is a set, $0,1$ are elements of $E$ and
$\oplus$ is a partial binary operation on $E$ \cite{gg02,gud120,gud220,hz12}. When $a\oplus b$ is defined, we say that $a\oplus b$ \textit{exists} and write
$a\perp b$. An effect algebra satisfies the following axioms:
\begin{list}
{(E\arabic{cond})}{\usecounter{cond}
\setlength{\rightmargin}{\leftmargin}}
\item If $a\perp b$, then $b\perp a$ and $a\oplus b=b\oplus a$.
\item If $a\perp b$, $c\perp (a\oplus b)$, then $b\perp c$, $a\perp (b\oplus c)$ and
$a\oplus (b\oplus c)=(a\oplus b)\oplus c$.
\item If $a\in E$, there exists a unique $a'\in E$ such that $a'\perp a$ and $a\oplus a'=1$.
\item If $a\perp 1$, then $a=0$.
\end{list}
It is easy to check that $\paren{\escript (H),0,I,\oplus}$ is an effect algebra where $a\perp b$ when $a+b\le I$ and in this case we define $a\oplus b=a+b$. We interpret the effect $a\oplus b$ to be a parallel stochastic sum of $a$ and $b$.

For $a,b\in\escript (H)$ we define their \textit{standard sequential product} by $a\circ b=a^{1/2}ba^{1/2}$ \cite{gg02,gn01}. We interpret $a\circ b$ as the effect that results from first measuring $a$ and then measuring $b$. In this way, the measurement of $a$ can interfere with the measurement of $b$ but not vice-versa. It is shown in \cite{gn01} that $a\circ b=b\circ a$ if and only if $ab=ba$. Most of the following properties of the standard sequential product are straightforward to show
\cite{gg02}.

\begin{lem}    
\label{lem21}
{\rm{(1)}}\enspace $a\circ (b\oplus c)=a\circ b\oplus a\circ c$.
{\rm{(2)}}\enspace $I\circ a=a\circ I=a$.
{\rm{(3)}}\enspace If $a\circ b=0$, then $ab=ba$.
{\rm{(4)}}\enspace  If $ab=ba$, then $a\circ (b\circ c)=(a\circ b)\circ c$ for all $c\in\escript (H)$.
{\rm{(5)}}\enspace If $ac=ca$ and $bc=cb$, then $c(a\circ b)=(a\circ b)c$ and $c(a\oplus b)=(a\oplus b)c$.
{\rm{(6)}}\enspace $a\circ b\le a$ for all $a,b\in\escript (H)$.
{\rm{(7)}}\enspace If $a\le b$, then $c\circ a\le c\circ b$ for all $c\in\escript (H)$.
\end{lem}

Notice that if $A=\brac{a_x\colon x\in\Omega _A}$ is a finite observable, then the corresponding L\"uders instrument has the form
\begin{equation*}
\lscript (\Delta )(\rho )=\sum _{x\in\Delta }a_x\circ\rho
\end{equation*}
and its determined sub-observables are given by
\begin{equation*}
\lscript _b^*(\Delta )=\sum _{x\in\Delta}a_x\circ b
\end{equation*}

\section{Sub-Observables}  
Let $A$ be a sub-observable that is not an observable and let $a=A(\Omega _A)\ne I$. For the outcome space $(\Omega _A,\fscript _A)$, let $y\notin\Omega _A$ and define $\Omega _B=\Omega _A\cup\brac{y}$, $\fscript _B=\fscript _A\cup\brac{\Delta\cup\brac{y}\colon\Delta\in\fscript _A}$. Then $\fscript _B$ is a $\sigma$-algebra of subsets of $\Omega _B$ and we call $(\Omega _B,\fscript _B)$ the \textit{one-point extension} of $(\Omega _A,\fscript _A)$. For $\Gamma\in\fscript _B$ define
$B(\Gamma )=A(\Gamma )$ if $\Gamma\in\fscript _A$ and $B(\Gamma )=A(\Delta )+a'$ if $\Gamma=\Delta\cup\brac{y}$, $\Delta\in\fscript _A$. Then $B$ is an observable with outcome space $(\Omega _B,\fscript _B)$ that we call the \textit{minimal extension} of $A$. Notice that if $A=\ob (H)$, then $B$ is essentially the same as $A$. For a simple example, if $a\in\escript (H)$ with $a\ne 0,I$, then $A_x=a$ is a sub-observable with outcome space $\brac{\brac{x},\brac{\emptyset ,\brac{x}}}$. The minimal extension of $A$ is $B=\brac{B_x,B_y}$ where $B_x=a$, $B_y=a'$ and the outcome space is
$\brac{\brac{x,y},\brac{\emptyset ,\brac{x},\brac{y},\brac{x,y}}}$.

A \textit{sub-instrument} $\iscript$ satisfies the conditions for an instrument except $\iscript (\Omega _\iscript )$ need not be a channel. Let the Kraus decomposition of $\iscript (\Omega _\iscript )$ be $\iscript (\Omega )(\rho )=\sum C_i\rho C_i^*$ \cite{hz12,kra83,nc00} where $C_i\in\lscript (H)$ and $\sum C_i^*C_i=D<I$. Notice that $D\in\escript (H)$. Let $(\Omega _\jscript ,\fscript _\jscript )$ be the one-point extension of $(\Omega _\jscript ,\fscript _\jscript )$ and define
$\jscript (\Delta )=\iscript (\Delta )$ if $\Delta\in\Omega _\jscript$ and
\begin{equation*}
\jscript \paren{\Delta\cup\brac{y}}(\rho )=\iscript (\Delta )(\rho )+(I-D)^{1/2}\rho (I-D)^{1/2}=\iscript (\Delta )(\rho )+D'\circ\rho
\end{equation*}
Then $\jscript$ is an instrument called the \textit{minimal extension} of $\iscript$.

Let $A\in\sob (H)$ with outcome space $(\Omega _A,\fscript _A)$ and let $A_1$ be the minimal extension of $A$ with outcome space
$(\Omega _{A_1},\fscript _{A_1})$. Let $\jscript _1$ be an instrument with outcome space $(\Omega _{A_1},\fscript _{A_1})$ such that $\jscripthat _1=A_1$. (Such an instrument exists although it need not be unique.) Then for all $\Delta\in\fscript _{A_1}$ we have
\begin{equation*}
\trace\sqbrac{\jscript _1(\Delta )(\rho )}=\trace\sqbrac{\rho A_1(\Delta )}
\end{equation*}
for all $\rho\in\sscript (H)$. Hence, for all $\Delta\in\fscript _A$ we have $\trace\sqbrac{\jscript _1(\Delta )(\rho )}=\trace\sqbrac{\rho A(\Delta )}$ for all
$\rho\in\sscript (H)$. For example, let $\iscript _a^*$ be a sub-observable determined by $\iscript\in\instr (H)$. Let $\iscript _{a,1}^*$ be the minimal extension of
$\iscript _a^*$ with outcome space $(\Omega _1,\fscript _1)$ and let $\jscript _1$ be an instrument with outcome space $(\Omega _1,\fscript _1)$ such that
$(\Omega _1,\fscript _1)$ and $\jscripthat _1=\iscript _{a,1}^*$. Then for every $\Delta\in\fscript _\iscript$ we have
\begin{equation*}
\trace\sqbrac{\jscript _1(\Delta )\rho}=\trace\sqbrac{\rho\iscript _a^*(\Delta )}
\end{equation*}

Let $A,B\in\sob (H)$ and let $A_1$ be the minimal extension of $A$. If $\iscript\in\instr (H)$ satisfies $\iscripthat =A_1$, then define the
$\iscript$-\textit{sequential product of} $A$ \textit{then} $B$ \textit{with outcome space} $(\Omega _A\times\Omega _B,\fscript _A\times\fscript _B)$ to be the sub-observable $A\sqbrac{\iscript}B=\iscript ^*(B)$. This is shorthand notation for
\begin{equation*}
A\sqbrac{\iscript}B(\Delta\times\Gamma )=\iscript ^*(B)(\Delta\times\Gamma )=\iscript ^*(\Delta )\paren{B(\Gamma )}
\end{equation*}
for all $\Delta\in\fscript _A$, $\Gamma\in\fscript _B$ \cite{gud120,gud220,gud21}. We also define $B$ $\iscript$-\textit{conditioned by} $A$ to be the sub-observable with \textit{outcome space} $(\Omega _B,\fscript _B)$ given by
\begin{equation*}
(B\mid\iscript\mid A)(\Gamma )=\iscriptbar\,^*\paren{B(\Gamma )}=A\sqbrac{\iscript}B(\Omega _A\times\Gamma )
\end{equation*}
for all $\Gamma\in\fscript _B$ \cite{gud120,gud220,gud21}.

\begin{exam}{1}  
Let $A=\brac{a_x\colon x\in\Omega}$ be a finite observable and let $\lscript$ be the L\"uders instrument given by $\lscript _x(\rho )=a_x\circ\rho$. Then for
$b\in\escript (H)$ we have $\lscript _b^*(\Delta )=\sum\limits _{x\in\Delta}a_x\circ b$ and we define $b_1\in\escript$ by
\begin{equation*}
b_1=\lscript _b^*(\Omega )=\sum _{x\in\Omega}a_x\circ b
\end{equation*}
Let $\Omega _1=\Omega\cup\brac{y}$ be the one-point extension of $\Omega$ and the minimal extension of $\lscript _b^*$ satisfies $\lscript _{b,1}^*(x)=a_x\circ b$ for $x\in\Omega$, $\lscript _{b,1}^*(y)=b'$ and
\begin{equation*}
\lscript _{b,1}^*(\Delta\cup\brac{y})=\sum _{x\in\Delta}a_x\circ b+b'_1
\end{equation*}
for $\Delta\subseteq\Omega$. Let $\jscript$ be the instrument on $\Omega _1$ given by
\begin{equation*}
\jscript (x)(\rho )=(a_x\circ b)^{1/2}\rho (a_x\circ b)^{1/2}=(a_x\circ b)\circ\rho
\end{equation*}
for $x\in\Omega$ and $\jscript (y)(\rho )=(b'_1)\circ\rho$. Then $\jscripthat =\lscript _{b,1}^*$ so $\jscript$ measures $\lscript _{b,1}^*$. The $\jscript$-sequential product of $\lscript _b^*$ then $\lscript _c^*$ becomes \cite{gud21}
\begin{equation*}
\lscript _b^*\sqbrac{\jscript}\lscript _c^*=\jscript ^*(\lscript _c^*)
\end{equation*}
We then obtain
\begin{align*}\lscript _b^*\sqbrac{\jscript}\lscript _c^*(x,z)&=\jscript ^*(\lscript _c^*)(x,z)=\jscript^*(x)\sqbrac{\lscript _c^*(z)}=\jscript^*(x)(a_z\circ c)\\
   &=(a_x\circ b)\circ (a_z\circ c)
\end{align*}
Moreover \cite{gud120,gud220},
\begin{align*}
\paren{\lscript _c^*\mid\jscript\mid\lscript _b^*}(z)&=\lscript _b^*\sqbrac{\jscript}(\Omega\times z)=\sum _{x\in\Omega}\jscript ^*(\lscript _c^*)(x,z)\\
   &=\sum _{x\in\Omega}(a_x\circ b)\circ (a_z\circ c)
\end{align*}
Another instrument that measures $\lscript _{b,1}^*$ is the finite Holevo instrument $\hscript _{(\alpha ,\lscript _{b,1}^*)}$ with states $\alpha _x$ and observable $A$. We then have
\begin{align*}
\lscript _b^*\sqbrac{\hscript _{(\alpha ,\lscript _{b,1}^*)}}\lscript _c^*(x,z)&=\hscript _{(\alpha ,\lscript _{b,1}^*)}(x)(\lscript _c^*)(z)
       =\hscript _{(\alpha ,\lscript _{b,1}^*)}(x)(a_z\circ c)\\
       &=\trace\sqbrac{\alpha _x(a_z\circ c)}\sqbrac{\lscript _{b,1}^*(x)}=\trace\sqbrac{\alpha _x(a_z\circ c)}a_x\circ b
\end{align*}
Moreover,
\begin{align*}
\paren{\lscript _c^*\mid\hscript _{(\alpha ,\lscript _{b,1}^*)}\mid\lscript _b^*}(z)
   &=\sum _{x\in\Omega}\lscript _b^*\sqbrac{\hscript _{(\alpha ,\lscript _{b,1}^*)}}\lscript _c^*(x,z)\\
   &=\sum _{x\in\Omega}\trace\sqbrac{\alpha _x(a_z\circ c)}a_x\circ b\hskip 10pc\square
\end{align*}
\end{exam}

\begin{exam}{2}  
Let $\hscript _{(\alpha ,A)}(\Delta )(\rho )=\trace\sqbrac{\rho A(\Delta )}\alpha$ be a Holevo instrument and let 
\begin{equation*}
(\hscript _{(\alpha ,A)}^*)_a(\Delta )=\trace (\alpha a)A(\Delta )
\end{equation*}
be a sub-observable determined by $\hscript _{(\alpha ,A)}$. The minimal extension of $(\hscript _{(\alpha ,A)}^*)_a$ satisfies
\begin{equation*}
D\paren{\Delta\cup\brac{y}}=(\hscript _{(\alpha ,A),1}^*)_a\paren{\Delta\cup\brac{y}}=\trace (\alpha a)A(\Delta )+\sqbrac{1-\trace (\alpha a)}I
\end{equation*}
Notice that $D$ is additive because if $\Delta\cap\Gamma =\emptyset$ and $y\notin\Delta\cup\Gamma$, then
\begin{align*}
D\paren{\Delta\cup\brac{y}\cup\Gamma}&=\trace (\alpha a)A(\Delta\cup\Gamma )+\sqbrac{1-\trace (\alpha a)}\\
   &=\trace (\alpha a)\sqbrac{A(\Delta )+A(\Gamma )}+\sqbrac{1-\trace (\alpha a)}I\\
   &=D\paren{\Delta\cup\brac{y}}+D(\Gamma )
\end{align*}
Since $\hscript _{(\beta ,D)}$ measures $D$ we have
\begin{align*}
(\hscript _{(\alpha ,A)}^*)_a\sqbrac{\hscript _{(\beta ,D)}}(\hscript _{(\alpha ,A)}^*)_b(\Delta\times\Gamma )
   &=\hscript _{(\beta ,D)}^*(\Delta )\sqbrac{(\hscript _{(\alpha ,A)}^*)_b(\Gamma )}\\
   &=\trace\sqbrac{\beta (\hscript _{(\alpha ,A)}^*)_b(\Gamma )}D(\Delta )\\
   &=\trace\sqbrac{\beta\trace (\alpha b)A(\Gamma )}\trace (\alpha a)A(\Delta )\\
   &=\trace (\alpha b)\trace (\alpha a)\trace\sqbrac{\beta A(\Gamma )}A(\Delta )
\end{align*}
Moreover,
\begin{align*}
(\hscript _{(\alpha ,A)}^*)_b\mid\hscript _{(\beta ,D)}\mid (\hscript _{\alpha ,A)}^*)_a(\Gamma )
   &=(\hscript _{(\alpha ,A)}^*)_a\sqbrac{\hscript _{(\beta D)}}(\hscript _{(\alpha ,A)}^*)_b(\Omega _A\times\Gamma )\\
   &=\trace (\alpha ,b)\trace (\alpha a)\trace\sqbrac{\beta A(\Gamma )}I\hskip 6pc\square
\end{align*}
\end{exam}

If $A\in\sob (H)$ and $\Delta _1,\Delta _2\in\fscript _A$ with $\Delta _1\subseteq\Delta _2$, then $A(\Delta _1)\le A(\Delta _2)$ and in particular,
$A(\Delta )\le A(\Omega _A)$ for all $\Delta\in\fscript _A$. For $A,B\in\sob (H)$ we write $A\le B$ if $(\Omega _A,\fscript _A)=(\Omega _B,\fscript _B)$ and
$A(\Delta )\le B(\Delta )$ for all $\Delta\in\fscript _A$. Let $\uscript\subseteq\sob (H)$ and assume that all elements of $\uscript$ have the same outcome set
$(\Omega ,\fscript )$. We call $\uscript$ a $\sob$ \textit{effect algebra} if:
\begin{list}
{(S\arabic{cond})}{\usecounter{cond}
\setlength{\rightmargin}{\leftmargin}}
\item there exists an observable $Z\in\uscript$,
\item if $A\in\uscript$ then $A'=Z-A\in\uscript$,
\item if $A,B\in\uscript$ and $A+B\in\sob (H)$, then $A+B\in\uscript$.
\end{list}
Notice that $0\in\uscript$ because $0=Z-Z\in\uscript$. Also, $Z=A$ is the only observable in $\uscript$. Indeed, suppose $A$ is an observable in $\uscript$, then
$A'=Z-A\in\uscript$. Since
\begin{equation*}
A'(\Omega )=Z(\Omega )-A(\Omega )=I-I=0
\end{equation*}
we conclude that $A'(\Delta )\le A'(\Omega )=0$ for all $\Delta\in\fscript$. Hence, $A'=0$ so $A=Z$. If $A,B\in\sob (H)$ with $A+B\in\sob (H)$, we write $A\perp B$. When $A\perp B$ we define $A\oplus B=A+B$ and say that $A\oplus B$ \textit{exists}.

\begin{thm}    
\label{thm31}
If $\uscript$ is a $\sob$ effect algebra, then $(\uscript, 0,Z,\oplus )$ is an effect algebra.
\end{thm}
\begin{proof}
There are four conditions to be satisfied which we now check.\newline
(E1)\enspace If $A,B\in\uscript$ with $A\perp B$, then $B\perp A$ and $A\oplus B=B\oplus A=A+B\in\uscript$.\newline
(E2)\enspace If $A,B,C\in\uscript$ with $A\perp B$ and $C\perp (A\oplus B)$, then $(A\oplus B)\oplus C=A+B+C\in\uscript$,
Hence, $B\perp C$ and $A\perp (B\oplus C)$ so
\begin{equation*}
A\oplus (B\oplus C)=A+B+C=(A\oplus B)\oplus C
\end{equation*}
(E3)\enspace If $A\in\uscript$, then $A'=Z-A$ is the unique element of $\uscript$ satisfying $A\oplus A'=Z$.\newline
(E4)\enspace If $A\in\uscript$ and $A\perp Z$, then
\begin{equation*}
Z-(A\oplus Z)=Z-(A+Z)=-A
\end{equation*}
Hence, $-A\in\uscript$ and it follows that $A=0$.
\end{proof}

If $\iscript\in\instr (H)$ and $\uscript _\iscript =\brac{\iscript _a^*\colon a\in\escript (H)}$, then we conjecture that $\uscript$ need not be a $\sob$ effect algebra. However, we can change the definition of $\perp$ so that $\uscript _\iscript$ becomes an effect algebra. We write $\iscript _a^*\perp\iscript _b^*$ is $a\perp b$ and if $a\perp b$ we define $\iscript _a^*\oplus\iscript _b^*=\iscript _{(a+b)}^*$.

\begin{thm}    
\label{thm32}
If $\iscript\in\instr (H)$, then $(\uscript _\iscript ,0,\iscript _I^*,\oplus )$ is an effect algebra and $F(a)=\iscript _a^*$ is a morphism from $\escript (H)$ onto
$\uscript _\iscript$. Moreover, $(\iscript _a^*)'=\iscript _{a'}^*$.
\end{thm}
\begin{proof}
If $a\perp b$, then
\begin{equation*}
\iscript _{a+b}^*(\Delta )=\iscript ^*(\Delta )(a+b)=\iscript ^*(\Delta )a+\iscript ^*(\Delta )b=\iscript _a^*(\Delta )+\iscript _b^*(\Delta )
\end{equation*}
for all $\Delta\in\fscript _\iscript$. Hence, $\iscript _a^*\oplus\iscript _b^*=\iscript _{a+b}^*=\iscript _b^*+\iscript _b^*$. We now check the four conditions for an effect algebra.\newline
(E1)\enspace If $\iscript _a^*,\iscript _b^*\in\uscript _\iscript$ with $\iscript _a^*\perp\iscript _b^*$ we have $\iscript _b^*\perp\iscript _a^*$ and
\begin{equation*}
\iscript _b^*\oplus\iscript _a^*=\iscript _b+\iscript _a=\iscript _a+\iscript _b=\iscript _a^*\oplus\iscript _b^*
\end{equation*}
(E2)\enspace If $\iscript _a^*,\iscript _b^*,\iscript _c^*\in\uscript _\iscript$ with $\iscript _a^*\perp\iscript _b^*$ and $\iscript _c^*\perp (\iscript _a^*\oplus\iscript _b^*)$ then $a\perp b$ and $c\perp (a\oplus b)$. Hence, $a+b+c\in\escript (H)$ so $b\perp c$ and $a\perp (b\oplus c)$. Hence, $\iscript _b^*\perp\iscript _c^*$ and
$\iscript _a^*\perp (\iscript _b^*\oplus\iscript _c^*)$ and we have
\begin{equation*}
\iscript _a^*\oplus (\iscript _b^*\oplus\iscript _c^*)=\iscript _a^*+\iscript _b^*+\iscript _c^*=(\iscript _a^*\oplus\iscript _b^*)\oplus\iscript _c^*
\end{equation*}
(E3)\enspace If $\iscript _a^*\in\uscript _\iscript$, then $\iscript _{a'}^*\perp\iscript _a^*$ and $\iscript _a^*\oplus\iscript _{a'}^*=\iscript _{a+a'}^*=\iscript _I^*$.
If $\iscript _a^*\oplus\iscript _b^*=\iscript _I^*$, then $\iscript _b^*=\iscript _I^*-\iscript _a^*=\iscript _{a'}^*$ so $(\iscript _a^*)'=\iscript _{a'}^*$ is unique.\newline
(E4)\enspace If $\iscript _a^*\perp\iscript _I^*$, then $a\perp I$ so $a=0$ and hence, $\iscript _a^*=0$.\newline
To show that $F$ is a morphism, if $a\perp b$ we have
\begin{equation*}
F(a\oplus b)=\iscript _{a\oplus b}^*=\iscript _a^*\oplus\iscript _b^*=F(a)\oplus F(b)
\end{equation*}
Moreover, $F(I)=\iscript _I^*$ so $F$ is a morphism.
\end{proof}

We now show that certain subsets $\uscript _\iscript\subseteq\sob (H)$ are $\sob$ effect algebras.

\begin{thm}    
\label{thm33}
If $\iscript =\hscript _{(\alpha ,A)}$ is a Holevo instrument, then $\uscript _\iscript$ is a $\sob$ effect algebra.
\end{thm}
\begin{proof}
We have that $\uscript _\iscript=\brac{\iscript _a^*\colon a\in\escript (H)}$. We now check the three conditions for a $\sob$ effect algebra.\newline
(1)\enspace $\iscript _I^*\in\uscript _\iscript$ and $\iscript _I^*\in\ob (H)$.\newline
(2)\enspace If $\iscript _a^*\in\uscript _\iscript$, then for all $\Delta\in\fscript _\iscript$ we obtain
\begin{equation*}
\iscript _I^*(\Delta )-\iscript _a^*(\Delta )=\iscript ^*(\Delta )(I)-\iscript ^*(\Delta )(a)=\iscript ^*(\Delta )(I-a)=\iscript ^*(\Delta )(a')=\iscript _{a'}^*(\Delta )
\end{equation*}
Hence, $\iscript _I^*-\iscript _a^*=\iscript _{a'}^*\in\uscript _\iscript$.\newline
(3)\enspace Let $\iscript _a^*,\iscript _b^*\in\uscript _\iscript$ and suppose $\iscript _a^*+\iscript _b^*\in\sob (H)$. Then for all $\Delta\in\fscript _\iscript$ we have
\begin{align*}
(\iscript _a^*+\iscript _b^*)(\Delta )&=\iscript _a^*(\Delta )+\iscript _b^*(\Delta )=\iscript ^*(\Delta )(a)+\iscript ^*(\Delta )(b)
    =\trace (\alpha a)A(\Delta )+\trace (\alpha b)A(\Delta )\\
    &=\sqbrac{\trace (\alpha a)+\trace (\alpha b)}A(\Delta )
\end{align*}
It follows that $0\le\trace (\alpha a)+\trace (\alpha b)\le 1$. If $C=\sqbrac{\trace (\alpha a)+\trace (\alpha b)}I$ we have that $C\in\escript (H)$ and
\begin{equation*}
\iscript _c^*(\Delta )=\iscript (\Delta )(c)=\trace (\alpha c)A(\Delta )=\sqbrac{\trace (\alpha a)+\trace (\alpha b)}A(\Delta )=(\iscript _a^*+\iscript _b^*)(\Delta )
\end{equation*}
Therefore, $\iscript _a^*+\iscript _b^*=\iscript _c^*\in\uscript _\iscript$.\newline
We conclude that $\uscript _\iscript$ is a $\sob$ effect algebra.
\end{proof}

\begin{thm}    
\label{thm34}
If $\iscript _\alpha (\Delta )(\rho )=\iscript (\Delta )(\alpha )$ is a constant state instrument, then $\uscript _{\iscript _\alpha}$ is a $\sob$ effect algebra if and only if for every $a,b\in\escript (H)$ satisfying $\trace\sqbrac{\iscriptbar (\alpha )(a+b)}\le 1$ there exists a $c\in\escript (H)$ such that
\begin{equation*}
\trace\sqbrac{\iscript (\Delta )(\alpha )c}=\trace\sqbrac{\iscript (\Delta )(\alpha )(a+b)}
\end{equation*}
for all $\Delta\in\fscript _\iscript$.
\end{thm}
\begin{proof}
We have that $(\iscript _\alpha ^*)_a(\Delta )=\trace\sqbrac{\iscript (\Delta )(\alpha )a}I$ for all $a\in\escript (H)$, $\Delta\in\fscript _\iscript$. The three conditions for
$\uscript _\iscript$ to be a $\sob$ effect algebra are the following:\newline
(1)\enspace $(\iscript _\alpha ^*)_I=\trace\sqbrac{\iscript (\Delta )(\alpha )}I\in\ob (H)$.\newline
(2)\enspace $(\iscript _\alpha ^*)_I-(\iscript _\alpha ^*)_a=(\iscript _\alpha ^*)_{a'}\in\uscript _{\iscript _\alpha}$.\newline
(3)\enspace Let $(\iscript _\alpha ^*)_a,(\iscript _\alpha ^*)_b\in\uscript _\alpha$ and suppose that $(\iscript _\alpha ^*)_a+(\iscript _\alpha ^*)_b\in\sob (H)$. We then have
\begin{align}                
\label{eq31}
\sqbrac{(\iscript _\alpha ^*)_a+(\iscript _\alpha ^*)_b}(\Delta )&=(\iscript _\alpha ^*)_a(\Delta )+(\iscript _\alpha ^*)_b(\Delta )
    =\iscript _\alpha ^*(\Delta )(a)+\iscript _\alpha ^*(\Delta )(b)\\
    &=\brac{\trace\sqbrac{\iscript (\Delta )(\alpha )a}+\trace\sqbrac{\iscript (\Delta )(\alpha )b}}=\trace\sqbrac{\iscript (\Delta )(\alpha )(a+b)}I\notag
\end{align}
and hence, $\trace\sqbrac{\iscript (\Delta )(\alpha )(a+b)}\le1$ for all $\Delta\in\fscript _\iscript$.\newline 
 It follows that $\sqbrac{\iscriptbar (\alpha )(a+b)}\le 1$. If the given condition holds, there exists a $c\in\escript (H)$ such that
\begin{equation*}
(\iscript _\alpha ^*)_c(\Delta )=\trace\sqbrac{\iscript (\Delta )(\alpha )c}I=\trace\sqbrac{\iscript (\Delta )(\alpha )}I
    =(\iscript _\alpha ^*)_a(\Delta )+(\iscript _\alpha ^*)_b(\Delta )
\end{equation*}
for all $\Delta\in\fscript _\iscript$ so that $(\iscript _\alpha ^*)_c=(\iscript _\alpha ^*)_a+(\iscript _\alpha ^*)_b$.

Hence, $\uscript _{\iscript _\alpha}$ is a $\sob$ effect algebra. Conversely, suppose $\uscript _{\iscript _\alpha}$ is a $\sob$ effect algebra and $a,b\in\escript (H)$ satisfy $\trace\sqbrac{\iscriptbar (\alpha )(a+b)}\le 1$. Then by \eqref{eq31}
\begin{equation*}
\sqbrac{(\iscript _\alpha ^*)_a+(\iscript _\alpha ^*)_b}(\Delta )\le I
\end{equation*}
for all $\Delta\in\fscript _\iscript$ so $(\iscript _\alpha ^*)_a+(\iscript _\alpha ^*)_b\in\sob (H)$. Therefore, there exists a $c\in\escript (H)$ such that
$(\iscript _\alpha ^*)_a+(\iscript _\alpha ^*)_b=(\iscript _\alpha ^*)_c$. Again by \eqref{eq31}
\begin{equation*}
\trace\sqbrac{\iscript (\Delta )(\alpha )c}=\trace\sqbrac{\iscript (\Delta )(\alpha )(a+b)}\qedhere
\end{equation*}
\end{proof}

\begin{exam}{3}  
If $\lscript _x(\rho )=a_x\circ\rho$ is an arbitrary L\"uders instrument, we do not know whether $\uscript _\lscript$ is a $\sob$ effect algebra. However, in the case where $a_x$ are projections (we then call $\lscript$ \textit{sharp}) we can show that it is. Indeed, if $\lscript _a^*+\lscript _b^*\in\sob (H)$ then
\begin{align*}
c=\sum_{x\in\Omega _\lscript}a_x(a+b)a_x&=\sum _{x\in\Omega _\lscript}\sqbrac{\lscript _x^*(a+b)}
    =\sum _{x\in\Omega _\lscript}\sqbrac{\lscript _x^*(a)+\lscript _x^*(b)}\\
    &=\sum _{x\in\Omega _\lscript}\sqbrac{(\lscript _a^*)_x+(\lscript _b^*)_x}\le I
\end{align*}
We conclude that $c\in\escript (H)$ and since $a_xa_y=\delta _{xy}a_x$ for every $x,y\in\Omega _\lscript$ we have
\begin{align*}
(\lscript _c^*)_y&=a_yca_y=a_y\sum _{x\in\Omega _\lscript}a_x(a+b)a_xa_y=a_y(a+b)a_y=a_yaa_y+a_yba_y\\
    &=(\lscript _a^*)_y+(\lscript _b^*)_y
\end{align*}
Hence, $\lscript _c^*=\lscript _a^*+\lscript _b^*$ so $\uscript _\lscript$ is a $\sob$ effect algebra.\hfill\qedsymbol
\end{exam}

A subset $\vscript\subseteq\sob (H)$ is \textit{convex} if $A_i\in\vscript$, $0\le\lambda _i\le 1$, $i=1,2,\ldots ,n$, $\sum _{i=1}^n\lambda _i=1$, implies
$\sum _{i=1}^n\lambda _iA_i\in\vscript$. A $\sob$ effect algebra need not be convex. For example $\vscript =\brac{\brac{0,0},\brac{0,I}}$ is a $\sob$ effect algebra that is not convex because
\begin{equation*}
\tfrac{1}{2}\brac{0,0}+\tfrac{1}{2}\brac{0,I}=\brac{0,\tfrac{1}{2}\,I}\notin\vscript
\end{equation*}

\begin{thm}    
\label{thm35}
A $\sob$ effect algebra $\uscript$ is convex if and only if $A\in\uscript$, $0\le\lambda\le 1$ imply $\lambda A\in\uscript$.
\end{thm}
\begin{proof}
Suppose $\uscript$ is convex, $A\in\uscript$ and $0\le\lambda\le 1$. We then have
\begin{equation*}
\lambda A=\lambda A+(1-\lambda )0\in\uscript
\end{equation*}
Conversely, suppose $\lambda A\in\uscript$ whenever $A\in\uscript$ and $0\le\lambda\le 1$. We need to show that if $A_i\in\uscript$, $0\le\lambda _i\le 1$,
$\sum _{i=1}^n\lambda _i=1$, then $\sum _{i=1}^n\lambda _iA_i\in\uscript$. We employ induction on $n$. If $n=2$, suppose that $A_1,A_2\in\uscript$,
$0\le\lambda _1,\lambda _2\le 1$ and $\lambda _1+\lambda _2=1$. By assumption $\lambda _1A_1,\lambda _2A_2\in\uscript$ and since
\begin{equation*}
\lambda _1A_1+\lambda _2A_2\le\lambda _1I+\lambda _2I=(\lambda _1+\lambda _2)I=I
\end{equation*}
we have that $\lambda _aA_1+\lambda _2A_2\in\sob (H)$. Since $\uscript$ is a $\sob$ effect algebra
$\lambda _1A_1+\lambda _2A_2\in\uscript$ so the result holds for $n=2$. Proceeding by induction, suppose the result holds for $n\ge 2$, $A_i\in\uscript$,
$i=1,2,\ldots ,n+1$, $0\le\lambda _i\le 1$ and $\sum _{i=1}^{n+1}\lambda _i=1$. Letting $\mu=\sum _{i=1}^n\lambda _i$ we can assume that $\mu\ne 0$. We have that $\sum _{i=1}^n\lambda _i/\mu =1$ so by hypothesis $\tfrac{1}{\mu}\sum _{n=1}^n\lambda _iA_i\in\uscript$. Since $0\le\mu\le 1$ we obtain
$\sum _{i=1}^n\lambda _iA_i\in\uscript$. Since $A_{n+1}\in\uscript$ we have that $\lambda _{n+1}A_{n+1}\in\uscript$. Moreover,
$\sum _{i=1}^n\lambda _iA_i+\lambda _{n+1}A_{n+1}\in\sob (H)$ so $\sum _{i=1}^{n+1}\lambda _iA_i\in\uscript$.
\end{proof}

Notice that $\uscript _\iscript =\brac{\iscript _a^*\colon a\in\escript (H)}$ is convex because if $\iscript _{a_i}^*\in\uscript _\iscript$ and $0\le\lambda _i\le 1$ with
$\sum\lambda _i=1$, then $b=\sum\lambda _ia_i\in\escript (H)$. We conclude that
\begin{equation*}
\sum\lambda _i\iscript _{a_i}^*=\iscript _b^*\in\uscript
\end{equation*}

If $a,b\in\escript (H)$ we define the \textit{sequential product} of $\iscript _a^*$ and $\iscript _b^*$ to be the sub-observable
$\iscript _a^*\circ\iscript _b^*(\Delta )=\iscript _{a\circ b}^*(\Delta )$ for all $\Delta\in\fscript _\iscript$. If $ab=ba$, it follows that
$\iscript _a^*\circ\iscript _b^*=\iscript _b^*\circ\iscript _a^*$. The next theorem follows from Lemma~\ref{lem21}.

\begin{thm}    
\label{thm36}
The sequential product $\iscript _a^*\circ\iscript _b^*$ satisfies the following conditions:
{\rm{(1)}}\enspace If $b\perp c$, then $\iscript _a^*\circ (\iscript _b^*+\iscript _c^*)=\iscript _a^*\circ\iscript _b^*+\iscript _a^*\circ\iscript _c^*$.\newline
{\rm{(2)}}\enspace If $0\le\lambda _i\le 1$, $\sum _{i=1}^n\lambda _i=1$, then for any $b_1,b_2,\ldots ,b_n\in\in\escript (H)$ we have
\begin{equation*}
\iscript _a^*\circ \paren{\sum _{i=1}^n\lambda _i\iscript _{b_i}^*}=\sum _{i=1}^n\lambda _i\iscript _a^*\circ\iscript _{b_i}^*
\end{equation*}
{\rm{(3)}}\enspace $\iscript _I^*\circ\iscript _a^*=\iscript _a^*\circ\iscript _I^*=\iscript _a^*$ for all $a\in\escript (H)$.\newline
{\rm{(4)}}\enspace If $a\circ b=0$, then $\iscript _a^*\circ\iscript _b^*=\iscript _b^*\circ\iscript _a^*$.\newline
{\rm{(5)}}\enspace If $ab=ba$, then $\iscript _a^*\circ (\iscript _b^*\circ\iscript _c^*)=(\iscript _a^*\circ\iscript _b^*)\circ\iscript _c^*$.\newline
{\rm{(6)}}\enspace If $ac=ca$ and $bc=cb$ then $\iscript _c^*\circ (\iscript _a^*\circ\iscript _b^*)=(\iscript _a^*\circ\iscript _b^*)\circ\iscript _c^*$ and\newline
$\iscript _c^*\circ (\iscript _a^*+\iscript _b^*)=(\iscript _a^*+\iscript _b^*)\circ\iscript _c^*$ when $a\perp b$.\newline
{\rm{(7)}}\enspace $\iscript _a^*\circ\iscript _b^*\le\iscript _a^*$\newline
{\rm{(8)}}\enspace If $a\le b$, then $\iscript _c^*\circ\iscript _a^*\le\iscript _c^*\circ\iscript _b^*$ for every $c\in\escript (H)$.
\end{thm}

The distribution of $\iscript _a^*\circ\iscript _b^*$ in the state $\rho$ becomes
\begin{align}                
\label{eq32}
\trace\sqbrac{\rho\iscript _a^*\circ\iscript _b(\Delta )}&=\trace\sqbrac{\rho\iscript _{a\circ b}^*(\Delta )}=\trace\sqbrac{\rho\iscript ^*(\Delta )(a\circ b)}
   =\trace\sqbrac{\iscript (\Delta )(\rho )a\circ b}\notag\\
   &=\trace\sqbrac{\paren{a\circ\iscript (\Delta )(\rho )}b}
\end{align}

\begin{exam}{4}  
For the Holevo instrument $\hscript _{(\alpha ,A)}$ we have
\begin{align*}
(\hscript _{(\alpha ,A)}^*)_a\circ (\hscript _{(\alpha ,A)}^*)_b&(\Delta )=(\hscript _{(\alpha ,A)}^*)_{a\circ b}(\Delta )=\trace (\alpha a\circ b)A(\Delta )\\
   &=\trace (\alpha a)\trace\sqbrac{(a\circ b)^\sim b}A(\Delta )=\trace\sqbrac{(a\circ b)^\sim b}(\hscript _{(\alpha ,A)})_a(\Delta )
\end{align*}
Not only is $(\hscript _{(\alpha ,A)}^*)_a\circ (\hscript _{(\alpha ,A)}^*)_b\le (\hscript _{(\alpha ,A)}^*)_a$ as in Theorem~\ref{thm36}(7) but
\begin{equation*}
(\hscript _{(\alpha ,A)}^*)_a\circ (\hscript _{(\alpha ,A)}^*)_b=\lambda (\hscript _{(\alpha ,A)}^*)_a
\end{equation*}
for a constant $\lambda$. Writing $\iscript =\hscript _{(\alpha ,A)}$, since $\iscript (\Delta )(\rho )=\trace\sqbrac{\rho A(\Delta )}$ it follows from \eqref{eq32} that the distribution of $\iscript _a^*\circ\iscript _b^*$ is
\begin{equation*}
\trace\sqbrac{\rho\iscript _a^*\circ\iscript _b^*(\Delta )}=\trace\sqbrac{(a\circ\alpha )b}\trace\sqbrac{\rho A(\Delta )}
    =\trace\sqbrac{(a\circ\alpha )b}\trace\sqbrac{\iscript (\Delta )(\rho )}
\end{equation*}
which is a constant times the distribution of $\iscript$. If $\iscript _\alpha (\Delta )=\iscript (\Delta )\alpha$ is a constant state instrument, we obtain
\begin{equation*}
(\iscript _\alpha ^*)_a\circ (\iscript _\alpha ^*)_b(\Delta )=(\iscript _\alpha ^*)_{a\circ b}(\Delta )=\trace\sqbrac{\iscript (\Delta )(\alpha )a\circ b}I
   =\trace\sqbrac{\paren{a\circ\iscript (\Delta )\alpha}b}I
\end{equation*}
In particular, if $\iscript =\hscript _{(\beta ,A)}$ is a Holevo instrument, then 
\begin{align*}
(\iscript )_\alpha ^*)_a(\Delta )&=\trace\sqbrac{\hscript _{(\beta ,A)}(\Delta )(\alpha )a}=\trace\sqbrac{\trace\paren{\alpha A(\Delta )}\beta a}I
   =\trace\sqbrac{\alpha A(\Delta )}\trace (\beta a)I
\intertext{and}
(\iscript _\alpha ^*)_a&\circ (\iscript _\alpha ^*)_b(\Delta )=\trace\sqbrac{\alpha A(\Delta )}\trace\paren{a\circ\beta b}I\hskip 12pc\square
\end{align*}
\end{exam}

\section{Sequential Products of Instruments}  
Let $\iscript,\jscript\in\instr (H)$ with outcome spaces $(\Omega _\iscript ,\fscript _\iscript )$, $(\Omega _\jscript ,\fscript _\jscript )$, respectively. We define the
\textit{sequential product of} $\iscript$ \textit{then} $\jscript$ to be the instrument \cite{gud21} with outcome space
$(\Omega _\iscript\times\Omega _\jscript ,\fscript _\iscript\times\fscript _\jscript )$ that satisfies
\begin{equation*}
(\iscript\circ\jscript )(\Delta\times\Gamma )(\rho )=\jscript (\Gamma )\sqbrac{\iscript (\Delta )(\rho )}
\end{equation*}
for all $\Delta\in\fscript _\iscript$, $\Gamma\in\fscript _\jscript$, $\rho\in\sscript (H)$. We also define $\jscript$ \textit{conditioned by} $\iscript$ to be the instrument give by \cite{gud21}.
\begin{equation*}
\paren{\jscript (\iscript )(\Gamma )}(\rho )=(\iscript\circ\jscript )(\Omega _\iscript\times\Gamma )(\rho )=\jscript (\Gamma )\sqbrac{\,\iscriptbar (\rho )}
\end{equation*}

\begin{thm}    
\label{thm41}
{\rm{(1)}}\enspace For all $a\in\escript (H)$ we have
\begin{equation*}
(\iscript\circ\jscript )_a^*(\Delta\times\Gamma )=\iscript ^*(\Delta )\sqbrac{\jscript _a^*(\Gamma )}=\iscript _{\jscript _a^*(\Gamma )}^*(\Delta )
\end{equation*}
{\rm{(2)}}\enspace The observable measured by $\iscript\circ\jscript$ satisfies
\begin{equation*}
(\iscript\circ\jscript )_I^*(\Delta\times\Gamma )=\iscript _{\jscripthat (\Gamma )}^*(\Delta )
\end{equation*}
{\rm{(3)}}\enspace For all $a\in\escript (H)$ we have
\begin{equation*}
(\jscript\mid\iscript )_a^*(\Gamma )=\iscript _{\jscript _a^*(\Gamma )}^*(\Omega _\iscript )
\end{equation*}
{\rm{(4)}}\enspace The observable measured by $(\jscript\mid\iscript )$ satisfies
\begin{equation*}
(\jscript\mid\iscript )_I^*(\Gamma )=\iscript _{\jscripthat (\Gamma )}^*(\Omega _\iscript )
\end{equation*}
\end{thm}
\begin{proof}
(1)\enspace For all $\rho\in\sscript (H)$, $a\in\escript (H)$ we have
\begin{align*}
\trace\sqbrac{\rho (\iscript\circ\jscript )^*(\Delta\times\Gamma )(a)}&=\trace\sqbrac{(\iscript\circ\jscript )(\Delta\times\Gamma )(\rho )a}
    =\trace\brac{\jscript (\Gamma )\sqbrac{\iscript (\Delta )(\rho )}a}\\
    &=\trace\sqbrac{\iscript (\Delta )(\rho )\jscript ^*(\Gamma )(a)}=\trace\brac{\rho\iscript ^*(\Delta )\sqbrac{\jscript ^*(\Gamma )(a)}}
\end{align*}
Hence,
\begin{equation*}
(\iscript\circ\jscript )^*(\Delta\times\Gamma )(a)=\iscript ^*(\Delta )\sqbrac{\jscript ^*(\Gamma )(a)}=\jscript ^*\circ\iscript ^*(\Gamma\times\Delta )(a)
\end{equation*}
It follows that
\begin{equation*}
(\iscript\circ\jscript )_a^*(\Delta\times\Gamma )=\iscript ^*(\Delta )\sqbrac{\jscript _a(\Gamma )}=\iscript _{\jscript _a^*(\Gamma )}^*(\Delta )
\end{equation*}
(2)\enspace It follows from (1) that
\begin{equation*}
(\iscript\circ\jscript )_I^*(\Delta\times\Gamma )=\iscript _{\jscript _I^*(\Gamma )}^*(\Delta )=\iscript _{\jscripthat (\Gamma )}^*(\Delta )
\end{equation*}
(3)\enspace Since
\begin{align*}
\trace\sqbrac{\rho (\jscript\mid\iscript )_a^*(\Gamma )}&=\trace\sqbrac{\rho (\jscript\mid\iscript )^*(\Gamma )(a)}=\trace\sqbrac{(\jscript\mid\iscript )(\Gamma )(\rho )a}\\
   &=\trace\sqbrac{\jscript (\Gamma )\sqbrac{\,\iscriptbar (\rho )}a}
  =\trace\sqbrac{\,\iscriptbar (\rho )\jscript ^*(\Gamma )(a)}\\
  &=\trace\sqbrac{\,\iscriptbar (\rho )\jscript _a^*(\Gamma )}
   =\trace\sqbrac{\iscript (\Omega _\iscript )(\rho )\jscript _a^*(\Gamma )}\\
  &=\trace\sqbrac{\rho\iscript ^*(\Omega _\iscript )\paren{\jscript _a^*(\Gamma )}}
  =\trace\sqbrac{\rho\iscript _{\jscript _a^*(\Gamma )}^*(\Omega _\iscript )}
\end{align*}
We conclude that
\begin{equation*}
(\jscript\mid\iscript )_a^*(\Gamma )=\iscript _{\jscript _a^*(\Gamma )}^*(\Omega _\iscript ))
\end{equation*}
(4)\enspace It follows from (3) that
\begin{equation*}
(\jscript\mid\iscript )_I^*(\Gamma )=\iscript _{\jscript _I^*(\Gamma )}^*(\Omega _\iscript )=\iscript _{\jscripthat (\Gamma )}^*(\Omega _\iscript )\qedhere
\end{equation*}
\end{proof}

\begin{exam}{5}  
Let $\iscript =\hscript _{(\alpha ,A)}$, $\jscript =\hscript _{(\beta ,B)}$ be Holevo instruments so that $\iscript _a^*(\Delta )=\trace (\alpha a)A(\Delta )$ and
$\jscript _a^*(\Gamma )=\trace (\beta a)B(\Gamma )$. We then obtain
\begin{align*}
(\iscript\circ\jscript )(\Delta\times\Gamma )(\rho )&=\jscript (\Gamma )\sqbrac{\iscript (\Delta )(\rho )}=\jscript (\Gamma )\sqbrac{\trace\paren{\rho A(\Delta )}\alpha}\\
   &=\trace\sqbrac{\rho A(\Delta )}\jscript (\Gamma )(\alpha )=\trace\sqbrac{\rho A(\Delta )}\trace\sqbrac{\alpha B(\Gamma )}\beta
\end{align*}
For all $a\in\escript (H)$ we have
\begin{align*}
(\iscript\circ\jscript )_a^*(\Delta\times\Gamma )&=(\iscript\circ\jscript )^*(\Delta\times\Gamma )(a)=\iscript ^*(\Delta )\sqbrac{\jscript ^*(\Gamma )(a)}\\
    &=\iscript ^*(\Delta )\sqbrac{\trace (\beta a)B(\Gamma )}=\trace (\beta a)\trace\sqbrac{\alpha B(\Gamma )}A(\Delta )
\end{align*}
The observable measured by $\iscript\circ\jscript$ satisfies
\begin{equation*}
(\iscript\circ\jscript )_I^*(\Delta\times\Gamma )=\trace\sqbrac{\alpha B(\Gamma )}A(\Delta )
\end{equation*}
The instrument $\jscript$ conditioned by $\iscript$ becomes
\begin{equation*}
(\jscript\mid\iscript )(\Gamma )(\rho )=\jscript (\Gamma )\sqbrac{\,\iscripthat (\rho )}=\jscript (\Gamma )(\alpha )=\trace\sqbrac{\alpha B(\Gamma )}\beta
\end{equation*}
We then obtain
\begin{align*}(\jscript\mid\iscript )_a^*(\Gamma )&=\iscript ^*(\Omega _\iscript )\sqbrac{\jscript _a^*(\Gamma )}
   =\iscript ^*(\Omega _\iscript )\sqbrac{\trace (\beta a)B(\Gamma )}\\
   &=\trace (\beta a)\iscript ^*(\Omega _\iscript )\sqbrac{B(\Gamma )}=\trace (\beta a)\trace\sqbrac{\alpha B(\Gamma )}I
\end{align*}
The observable measured by $(\jscript\mid\iscript )$ is $(\jscript\mid\iscript )_I^*(\Gamma )=\trace\sqbrac{\alpha B(\Gamma )}I\hfill\square$
\end{exam}

\begin{exam}{6}  
Let $\iscript _x(\rho )=a_x\circ\rho$, $\jscript _y(\rho )=b_y\circ \rho$ be L\"uders instruments. We then have
\begin{equation*}
(\iscript\circ\jscript )(x,y)(\rho )=\jscript _y\sqbrac{\iscript _x(\rho )}=b_y\circ (a_x\circ\rho )
\end{equation*}
The dual instruments satisfying $\iscript _x^*(a)=a_x\circ a$, $\jscript _y^*(a)=b_y\circ a$ and we obtain
\begin{equation*}
(\iscript\circ\jscript )_a^*(x,y)=(\iscript\circ\jscript )^*(x,y)(a)=\jscript ^*\circ\iscript ^*(y,x)(a)=\iscript _x^*\sqbrac{\jscript _y^*(a)}=a_x\circ (b_y\circ a)
\end{equation*}
The observable measured by $\iscript\circ\jscript$ becomes
\begin{equation*}
(\iscript\circ\jscript )_I^*(x,y)=a_x\circ b_y
\end{equation*}
which is the standard sequential product of the observable $A=\brac{a_x\colon x\in\Omega _A}$ and $B=\brac{b_y\colon y\in\Omega _B}$ \cite{gud220,gud21}.
The instrument $\jscript$ conditioned by $\iscript$ becomes
\begin{equation*}
(\jscript\mid\iscript )_y(\rho )=\jscript _y\sqbrac{\,\iscriptbar (\rho )}=b_y\circ\sum _{x\in\Omega _\iscript}a_x\circ\rho =\sum _{x\in\Omega _\iscript}b_y\circ (a_x\circ\rho )
\end{equation*}
We then obtain
\begin{equation*}
(\jscript\mid\iscript )_a^*(y)=\iscript ^*(\Omega _\iscript )\sqbrac{\jscript _a^*(y)}=\sum _{x\in\Omega _\iscript}a_x\circ (b_y\circ a)
\end{equation*}
The observable measured by $(\jscript\mid\iscript$ is
\begin{equation*}
(\jscript\mid\iscript )_I^*=\sum _{x\in\Omega _\iscript}a_x\circ b_y
\end{equation*}
which is again of the standard form \cite{gud220,gud21}.\hskip 12pc$\square$
\end{exam}

\begin{exam}{7}  
Let $\iscript _\alpha$, $\jscript _\beta$ be constant-state instruments so that 
\begin{align*}
(\iscript _\alpha ^*)(\Delta )&=\iscript _\alpha ^*(\Delta )(a)=\trace\sqbrac{\iscript (\Delta )(\alpha )a}I\\
(\jscript _\beta ^*(\Gamma )&=\jscript _\beta ^*(\Gamma )(a)=\trace\sqbrac{\jscript (\Gamma )(\beta )a}I
\end{align*}
The sequential product becomes
\begin{align*}
(\iscript _\alpha\circ\jscript _\beta )(\Delta\times\Gamma )(\rho )&=\jscript _\beta (\Gamma )\sqbrac{\iscript _\alpha (\Delta )(\rho )}
   =\jscript _\beta (\Gamma )\sqbrac{\iscript (\Delta )(\alpha )}\\
   &=\jscript _\beta (\Gamma )\trace\sqbrac{\iscript (\Delta )(\alpha )}\sqbrac{\iscript (\Delta )(\alpha )^\sim}
   =\trace\sqbrac{\iscript (\Delta )(\alpha )}\jscript (\Gamma )(\beta )
\end{align*}
The sub-observables determined by $\iscript _\alpha\circ\jscript _\beta$ are given by
\begin{align*}
(\iscript _\alpha\circ\jscript _\beta )_a^*(\Delta\times\Gamma )&=\iscript _\alpha ^*(\Delta )\sqbrac{\jscript _\beta ^*(\Gamma )(a)}
   =\iscript _\alpha ^*(\Delta )\brac{\trace\sqbrac{\jscript (\Gamma )(\beta )a}I}\\
   &=\trace\sqbrac{\jscript (\Gamma )(\beta )a}\iscript _\alpha ^*(\Delta )(I)=\trace\sqbrac{\jscript (\Gamma )(\beta )a}\trace\sqbrac{\iscript (\Delta )\alpha }I
\end{align*}
The observable measured by $\iscript _\alpha\circ\jscript _\beta$ is
\begin{equation*}
(\iscript _\alpha\circ\jscript _\beta )_I^*(\Delta\times\Gamma )=\trace\sqbrac{\jscript (\Gamma )(\beta )}\trace\sqbrac{\iscript (\Delta )(\alpha )}I
\end{equation*}
The instrument $\jscript _\beta$ conditioned by $\iscript _\alpha$ satisfies
\begin{equation*}
(\jscript _\beta\mid\iscript _\alpha )(\Gamma )(\rho )=\jscript _\beta (\Gamma )\sqbrac{\,\iscriptbar _\alpha (\rho )}
   =\jscript _\beta (\Gamma )\sqbrac{\,\iscriptbar (\alpha )}=\jscript (\Gamma )(\beta )
\end{equation*}
so we conclude that $(\jscript _\beta\mid\iscript _\alpha )=\jscript _\beta$. We have that
\begin{equation*}
(\jscript _\beta\mid\iscript _\alpha )_a^*(\Gamma )=(\jscript _\beta ^*)_a(\Gamma )=\trace\sqbrac{\jscript (\Gamma )(\beta )a}I
\end{equation*}
and the observable measured by $(\jscript _\beta\mid\iscript _\alpha )$ becomes
\begin{equation*}
(\jscript _\beta\mid\iscript _\alpha )_I^*(\Gamma )=\trace\sqbrac{\jscript (\Gamma )(\beta }I\hskip 15pc\square
\end{equation*}
\end{exam}

So far we have considered examples of sequential products for two instruments of the same types. We now discuss sequential products of instruments of different types.

\begin{exam}{8}  
Let $\iscript$ be the L\"uders instrument $\iscript _x(\rho )=a_x\circ\rho$ and $\jscript$ the finite Holevo instrument $\jscript _y(\rho )=\trace\sqbrac{\rho B_y}\beta _y$. The sequential product becomes
\begin{equation*}
(\iscript\circ\jscript )(x,y)(\rho )=\jscript _y\sqbrac{\iscript _x(\rho )}=\jscript _y(a_x\circ\rho )=\trace\sqbrac{(a_x\circ\rho )B_y}\beta _y
\end{equation*}
The sub-observables determined by $\iscript\circ\jscript$ are given by
\begin{align*}
(\iscript\circ\jscript )_a^*(x,y)&=\iscript _x^*\sqbrac{\jscript _y^*(a)}=\iscript _x^*\sqbrac{\trace (\beta _ya)B_y}=\trace (\beta _ya)\iscript _x^*(B_y)\\
   &=\trace (\beta _ya)a_x\circ B_y
\end{align*}
As in Example~6, we obtain
\begin{equation*}
(\iscript\circ\jscript )_I^*(x,y)=a_x\circ B_y
\end{equation*}
which is the standard sequential product of the observables $A=\brac{a_x\colon x\in\Omega _A}$ and $B$. Letting
\begin{equation*}
C_y=(B\mid A)_y=\sum _x(a_x\circ B_y)
\end{equation*}
we obtain
\begin{align*}
(\jscript\mid\iscript )_y(\rho )&=\jscript _y\sqbrac{\,\iscriptbar (\rho )}=\jscript _y\sqbrac{\sum _{x\in\Omega _A}a_x\circ\rho}
   =\trace\sqbrac{\paren{\sum _{x\in\Omega _A}a_x\circ\rho}B_y}\beta _y\\
   &=\trace\sqbrac{\rho\sum _{x\in\Omega _A}(a_x\circ B_y)}\beta _y=\hscript _{(\beta ,C)}(y)(\rho )
\end{align*}
We have that
\begin{align*}
(\jscript\mid\iscript )_a^*(y)&=\iscript ^*(\Omega _\jscript )\sqbrac{\jscript _a^*(y)}=\sum _{x\in\Omega _\iscript}a_x\circ\sqbrac{\jscript _a^*(y)}
    =\sum _{x\in\Omega _\iscript}a_x\circ\sqbrac{\trace (\beta _ya)B_y}\\
    &=\trace (\beta _ya)\sum _{x\in\Omega _\iscript}(a_x\circ B_y)=\trace (\beta _ya)C_y
\end{align*}
The observable measured by $(\jscript\mid\iscript )$ becomes
\begin{equation*}
(\jscript\mid\iscript )_I^*(y)=\sum _{x\in\Omega _\iscript}(a_x\circ B_y)=C_y=(B\mid A)_y
\end{equation*}
We now consider the other order. These are given by the following equations.
\begin{align*}
(\jscript\circ\iscript )(y,x)(\rho )&=\iscript _x\sqbrac{\jscript _y (\rho )}=\iscript _x\sqbrac{\trace (\rho B_y)\beta _y}=\trace (\rho B_y)\iscript _x(\beta _y)\\
   &=\trace (\rho B_y)a_x\circ B_y\\
   (\jscript\circ\iscript )_a^*(x,y)&=\jscript _y^*\sqbrac{\iscript _x^*(a)}=\jscript _y^*(a_x\circ a)=\trace (\beta _ya_x\circ a)B_y\\
   (\jscript\circ\iscript )_I^*(x,y)&=\trace (\beta _ya_x)B_y\\
   (\iscript\mid\jscript )_x(\rho )&=\iscript _x\sqbrac{\,\jscriptbar (\rho )}=a_x\circ\sum _{y\in\Omega _\jscript}\jscript _y(\rho )
   =a_x\circ\sum _{y\in\Omega _\jscript}\trace (\rho B_y)\beta _y\\
   &=\sum _{y\in\Omega _\jscript}\trace (\rho B_y)a_x\circ\beta _y\\
   (\iscript\mid\jscript )_a^*(x)&=\jscript ^*(\Omega _\jscript )\sqbrac{\iscript _a^*(x)}=\sum _{y\in\Omega _\jscript}\jscript _y^*\sqbrac{\iscript _a^*(y)}
   =\sum _{y\in\Omega _\jscript}\trace\sqbrac{\beta _y\iscript _a^*(x)}B_y\\
   (\iscript\mid\jscript )_I^*(x)&=\sum _{y\in\Omega _\jscript}\trace (\beta _ya_x)B_y\hskip 16pc\square
\end{align*}
\end{exam}

\end{document}